\documentclass{llncs}
\usepackage{fullpage}
\usepackage{graphicx}
\usepackage{cite}
\usepackage{epsf}
\usepackage{epsfig}
\usepackage{epstopdf}
\usepackage{amsfonts}
\usepackage{amssymb}
\usepackage{amsmath}
\usepackage{latexsym}
\usepackage[all]{xy}

\usepackage{algorithmic}
\usepackage{algorithm}
\usepackage{hyperref}
\usepackage{nameref}
\begin{document}
\title{$2$-CLUB is \textsf{NP}-hard for distance to $2$-club cluster graphs}
\author{Mithilesh Kumar}
\institute{
  Simula@UiB\\
  Norway\\
  \email{mithilesh@simula.no}
}
\maketitle
\begin{abstract}
We show that $2$-CLUB is \textsf{NP}-hard for distance to $2$-club cluster graphs.
\end{abstract}
\section{Introduction}
A complete graph or clique is a graph that contains an edge for every pair of distinct vertices. Diameter of a graph is the length of a longest shortest path in the graph. Any clique has diameter $1$. A generalization of this notion is $s$-club, a graph of diameter $s$. In general graphs, finding a set of vertices that induces a subgraph of diameter $s$ is \textsf{NP}-hard. For $s=2$, Hartung et al \cite{HartungKN13} have studied the problem with many structral restrictions on the input graph. This paper answers one of the open problems mentioned in \cite{HartungKN13}.

Given a class of graphs with some property $\Pi$, we can define another class of graphs by the parameter distance to $\Pi$, namely the number of vertices that needs to be deleted from the graph to make the resultant graph have property $\Pi$. For example, distance $2$ to bipartiteness defines a class of graphs that become bipartite after deleting at most $2$ vectices. A graph where each connected component is an $s$-club is called $s$-club cluster graph. In this paper, we show that finding $2$-club in distance $d$ to $2$-club cluster graphs is \textsf{NP}-hard for $d\geq 2$.

\section{Constant Distance to 2-club cluster}
We define the $2$-CLUB problem as follows: Given an undirected graph $G = (V,E)$ and $k\in \mathbb{N}$, is there a vertex set $S\subseteq V$ of size at least $k$ such that $G[S]$ has diameter at most $2$?
\begin{theorem}
 2-CLUB is \textsf{NP}-hard even on graphs with distance two to $2$-club cluster.
\end{theorem}
\begin{proof}
 We reduce from the \textsf{NP}-hard CLIQUE problem: Given a positive integer $k$ and a graph $H$, the question is whether there is a clique 
 of size at least $k$.

 Given an instance $(H,k)$ of CLIQUE, we construct an undirected graph $G=(V,E)$. 

Let $|V(H)|=n$. Define the vertex set $$V(G):= V(H)\cup A \cup \{a,b,u\}\cup X_1\cup X_2$$ where $a, b, u$ are vertices and $A, B, X_1, X_2$ are sets of vertices with sizes $|X_1|=n^3, |X_2|=n^2-n$ and $|A|=n^2$. For every vertex $v_i\in V(H)$, we lebel $n$ vertices of $A$ as $V_i=\{v_{i,1},...,v_{i,n}\}$. 
The edge set $E(G)$ is defined as 
$$E(G)=E(H)\bigcup a\times \{\{b\}\cup X_1\cup A\}\bigcup b\times \{X_1\cup V(H)\cup X_2\}\bigcup V(H)\times X_2\bigcup u\times \{A\cup V(H) \cup X_2\}\bigcup_{ \forall i\in [n]} v_i\times V_i $$ 
Note that all the edges are undirected. See Figure below.
 \begin{figure}
  \centering
   \includegraphics[width=0.35\textwidth]{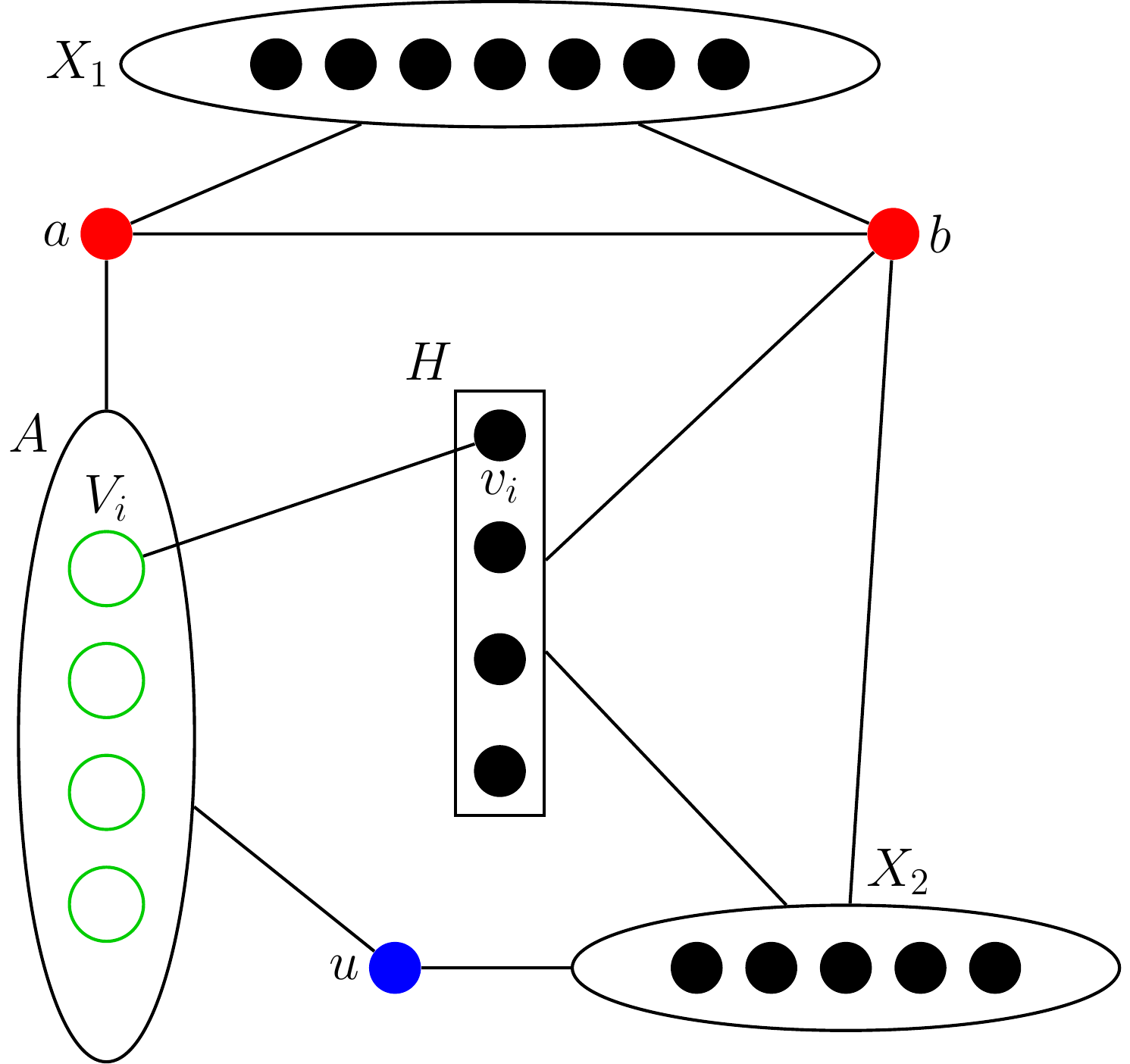}
 \end{figure}\\
 \begin{claim}
  $H$ has a clique of size $k$ if and only if $G$ has a 2-club of size $n^3+n^2+(k-1)n+k+2$.
 \end{claim}
 \begin{proof}
 Let $S$ be a clique of size $k$ in $H$. Then, $X_1\cup X_2 \cup S\cup (N(S)\cap A) \cup \{a,b\}$ is a 2-club of size $n^3+n^2+(k-1)n+k+2$.

Let $Y$ be a $2$-club of size $n^3+n^2+(k-1)n+k+2$ in $G$.

By size consideration $X_1\subset Y$. If $b\notin Y$, then none of $B$ and $X_2$ can be in $Y$. Consequently, the size of any $2$-club in $G$ can be $n^3+n^2+1$. Hence, we must have that $b\in Y$. By similar reasoning, we have that $a\in Y$. 

If $A\cap Y= \emptyset$, then the size of the largest 2-club can be at most $n^3+n^2+2$ implying that $Y$ must intersect with $A$. Moreover, $|A\cap Y|$ must be a multiple of $n$ as for $v_i\in A$ contained in $Y$, the whole subset $V_i\subset A$ can be included in $Y$ preserving the $2$-club property. If $|A\cap Y|<(k-1)n$, then size of the maximum 2-club can be at most $n^3+n^2+(k-1)n+2$, the size of $X_1\cup X_2\cup V(H)\cup \{a,b\}\cup (A\cap Y)$ which is less than $n^3+n^2+(k-1)n+k+2$. Hence at least $k$ vertices in $V(H)\cap Y$ have neighbors in $A\cap Y$. This also implies that $V(H)\cap Y$ forms a clique in $H$. If $\{x,y\}\in V(H)\cap Y$ are not adjacent and have neighbors $\{x',y'\}\in A\cap Y$. Then, there is no path of length $\leq 2$ between $x$ and $y'$. Hence, $H$ has a clique of size $k$.
\end{proof}

\end{proof}
\paragraph{Acknowledgements}
I am grateful for the fruitful discussions with Daniel Lokshtanov and Markus Dregi.
\bibliographystyle{alpha}	
\bibliography{biblio}
\end{document}